\documentclass[a4paper,10pt]{article}
\usepackage[utf8]{inputenc}
\usepackage[a4paper,left=3cm,right=2cm,top=2.5cm,bottom=2.5cm]{geometry}
\usepackage{graphicx}
\usepackage{epsf}
\usepackage{amssymb,amsmath,eufrak,amsthm,amscd,mathtools}
\usepackage{graphics}
\usepackage{textcomp}
\usepackage{amscd}
\usepackage{comment}
\usepackage{enumitem}
\usepackage{verbatim}
\usepackage{csquotes, nicefrac}
\usepackage{authblk}
\usepackage{capt-of,caption}
\usepackage{graphicx,newfloat,float}

\DeclareCaptionType{ctest}[Test]
\floatstyle{ruled}
\restylefloat{figure}

\newtheorem{theorem}{Theorem}[section]
\newtheorem{lemma}[theorem]{Lemma}

\newtheorem*{statement*}{Statement}
\newtheorem*{theorem*}{Theorem}
\newtheorem*{lemma*}{Lemma}
\newtheorem*{fact*}{Fact}

\theoremstyle{definition}
\newtheorem{definition}[theorem]{Definition}
\newtheorem*{definition*}{Definition}

\newtheorem*{example*}{Example}

\newtheorem*{exercise*}{Exercise}
\newtheorem{proposition}[theorem]{Proposition}
\newtheorem*{proposition*}{Proposition}
\newtheorem{corollary}[theorem]{Corollary}
\newtheorem*{corollary*}{Corollary}
\newtheorem{claim}[theorem]{Claim}
\newtheorem*{claim*}{Claim}

\newtheorem*{test*}{Test}

\theoremstyle{remark}
\newtheorem{remark}[theorem]{Remark}
\newtheorem*{remark*}{Remark}

\makeatletter
\newcommand{\l@abcd}[2]{\hbox to\textwidth{#1\dotfill #2}}
\newcommand{\F}{\mathbb{F}_2}
\newcommand{\xor}{\oplus}
\newcommand{\e}{\varepsilon}

\newcommand{\set}[1]{\left\{#1 \right\}}

\newcommand{\etalchar}[1]{$^{#1}$}

\makeatother

\def\dist{{\rm dist}}
\def\S{{\cal S}}

\title{Direct Sum Testing: \\ The General Case}
\author[1]{Irit Dinur}
\author[2]{Konstantin Golubev}
\affil[1]{The Weizmann Institute of Science, Israel}
\affil[2]{ETH Zurich, Switzerland}

\begin{document}

\maketitle

\abstract{
A function $f:[n_1]\times\dots\times[n_d]\to\F$  is a direct sum if it is of the form $f\left(a_1,\dots,a_d\right) = f_1(a_1)\xor\dots \xor f_d (a_d),$ for some $d$ functions $f_i:[n_i]\to\F$ for all $i=1,\dots, d$, and where $n_1,\dots,n_d\in\mathbb{N}$. We present a $4$-query test which distinguishes between direct sums and functions that are far from them. The test relies on the BLR linearity test (Blum, Luby, Rubinfeld, 1993) and on an agreement test which slightly generalizes the direct product test (Dinur, Steurer, 2014).

In multiplicative $\pm 1$ notation, our result reads as follows. A $d$-dimensional tensor with $\pm 1$ entries is called a tensor product if it is a tensor product of $d$ vectors with $\pm 1$ entries, or equivalently, if it is of rank $1$. The presented tests can be read as tests for distinguishing between tensor products and tensors that are far from being tensor products.

We also present a different test, which queries the function at most $(d+2)$ times, but is easier to analyze.
}

\section{Introduction}

Let us first fix some notations and definitions. By $[n]$ we mean the set $\{0,1,2,\dots, n\}$. For $d$ positive integers $n_1,\dots,n_d$, we denote $[\overline{n};d] = [n_1]\times\dots\times[n_d]$. For two functions $F, G: X\to Y$, we denote by $\dist(F,G)$ the relative Hamming distance between them, namely $\dist(F,G) = \Pr_{x\in X}[ F(x)\neq G(x)]$.  We say that $F:X\to Y$ is $\e$-close to have some Property, if there exists a function $G: X\to Y$ such that $g$ has the Property and $\dist(F,G)\leq \e$.

Given $d$ functions $f_i:[n_i]\to\F,\, i=1,\dots,d$, where $n_1,\dots,n_d\in\mathbb{N}$, their direct sum is the function $f:[\overline{n};d]\to\F$ given by $f\left(a_1,\dots,a_d\right) = f_1(a_1) \xor f_2(a_2) \xor \ldots \xor f_d(a_d)$, where $\xor$ stands for  addition is in the field $\F$. We denote $f = f_1 \oplus \cdots \oplus f_d$. We study the testability question: given a function $f: [\overline{n};d]\to\F$ test if it is a direct sum, namely if it belongs to the set
\[ 
DirectSum_{[\overline{n};d]} = \left\{ f_1\oplus\cdots \oplus f_d\;|\; f_i:[n_i]\to\F,\, i=1,\dots, d \right\}.
\]

Direct sum is a natural construction that is often used in complexity for hardness amplification \cite{Yao82, ImpagliazzoJK06, ImpagliazzoJKW08, Sudan2001, Trevisan2003}. It is related to the direct product construction: a function $f:[\overline{n};d]\to\F^d$ is the direct product of $f_1,\ldots,f_d$ as above if $f\left(a_1,\dots,a_d\right) = (f_1(a_1),\ldots,f_d(a_d))$ for all $(a_1,\ldots,a_d)\in[\overline{n};d]$. The testability of direct products has received attention \cite{GolSaf97,DR06,DG08,ImpagliazzoKW12,Dinur2014} as abstraction of certain PCP tests. It was not surprising to find \cite{David2015} that there is a connection between testing direct products to testing direct sum. However, somewhat unsatisfyingly this connection was confined to testing a certain type of {\em symmetric} direct sum. A symmetric direct sum is a function $f:[n]^d\to\F$ that is a direct product with all components equal; namely such that there is a single $g : [n]\to\F$ such that
\[f\left(a_1,\dots,a_d\right) = g(a_1)\xor g(a_2)\xor \cdots \xor g(a_d).\]
In \cite{David2015}, a 3-query test was presented for testing if a given $f$ is a symmetric direct sum, and the analysis carried out relying on the direct product test. It was left as an open question to devise and analyze a test for the property of being a (not necessarily symmetric) direct sum.

We design and analyze a four-query test which we call the ``square in a cube'' test, and show that it is a strong absolute local test for being a direct sum. That is, the number of queries is an absolute constant (namely, $4$), and the distance from a function to the subspace of direct sums is bounded by some absolute constant (independent of $n$ and $d$) times the probability of the failure of the test on this function. 
We also describe a simpler $(d+1)$-query test, whose easy analysis we defer to section~\ref{sec:shapkatest}.

In order to define the test, we need to introduce the following notation. Given two strings $a,b\in [\overline{n};d]$ and a set $S\subseteq [d]$, denote by $a_S b$ the string in $[\overline{n};d]$ whose $i$-th coordinate equals $a_i$ if $i\in S$ and $b_i$ otherwise.

\begin{figure}[h]\captionof{ctest}{Square in a Cube test. Given a query access to a function $f: [\overline{n};d]\to\F$:}\label{test:SiC-intro}
\begin{enumerate}
\item Choose $a,b\in [\overline{n};d]$ uniformly at random.
\item Choose two subsets $S,T\subset [d]$  uniformly at random, and let $U = S\triangle T$ be their symmetric difference.
\item Accept iff  \[f(a) \xor f(a_Sb) \xor f(a_Tb) \xor f(a_Ub) = 0.\]
\end{enumerate}
\end{figure}

We prove the following theorem for Test~\ref{test:SiC-intro}.
\begin{theorem}[Main]\label{thm:main}
There exists an absolute constant $c>0$ s.t. for all $d\in \mathbb{N}$ and $n_1,\dots,n_d\in \mathbb{N}$, given $f:[\overline{n};d]\to\F$,
\[ \dist (f, DirectSum_{[\overline{n};d]}) \le c\cdot \Pr_{a,b,S,T}[f(a) \xor f(a_Sb) \xor f(a_Tb) \xor f(a_{S\triangle T}b) \neq 0]
\]
\noindent where $a,b$ are chosen independently and uniformly from the domain of $f$, and $S,T$ are random subsets of $[d]$.
\end{theorem}

Our proof, similarly to \cite{David2015}, relies on a combination of the BLR linearity testing theorem \cite{Blum1993} and a direct product test, similar to the one analyzed in \cite{Dinur2014}. These two components were also used in the proof of  \cite{David2015} for the symmetric case, but here we use the components differently. The trick is to find the right combination. We first observe that once we fix $a,b$, the test is confined to a set of at most $2^d$ points in the domain, and can be viewed as performing a BLR (affinity rather than linearity) test on this piece of the domain. From the BLR theorem, we deduce an affine linear function on this piece. The next step is to combine the different affine linear functions, one from each piece, into one global direct sum, and this is done by reducing to direct product.

\paragraph{Testing if a tensor has rank $1$.}
An equivalent way to formulate our question is as a test for whether a $d$-dimensional tensor with $\pm 1$ entries has rank $1$. Indeed moving to multiplicative notation and writing $h_i = (-1)^{f_i}$ and $h = (-1)^f$,  we are asking whether there are $h_1,\ldots,h_d$ such that
\[ h = h_1\otimes \cdots\otimes h_d.
\]
Denoting
\[ TensorProduct_{[\overline{n};d]} = \left\{ h_1\otimes\cdots \otimes h_d\;|\; h_i:[n_i]\to\{-1,1\},
, i=1,\dots,d\right\}\]
we have
\begin{corollary}
There exists an absolute constant $c>0$ s.t. for all $d\in \mathbb{N}$ and $n_1,\dots,n_d\in \mathbb{N}$ , for every
$h:[\overline{n};d]\to\{-1,1\}$,
\[ \dist (h, TensorProduct_{[\overline{n};d]}) \le c\cdot \Pr_{a,b,S,T}[ h(a)\cdot h(a_Sb) \cdot h(a_Tb)\cdot h(a_{S\triangle T}b) \neq 1].
\]
\end{corollary}

\paragraph{Structure of the Paper.} 
In Sections~\ref{sec:SiC-test} and~\ref{sec:shapkatest} we present two different approaches for testing whether a $d$-dimensional binary tensor is a tensor product. In Section~\ref{sec:FurDirec} we discuss possible directions for future research. In Section~\ref{sec:DP}, we explain how to derive the specific direct product test that we need from the agreement testing theorem of \cite{DD}. This is used in the course of the proof in Section~\ref{sec:SiC-test}. The numbering is section-wise.
Finally, in Section~\ref{sec:FurDirec} we discuss possible directions for future research. 

\section{Square in a Cube Test}\label{sec:SiC-test}
In this section we present the Square in a Cube Test. Then we introduce the required background: the BLR test for a function being Affine in Subsection~\ref{subsec:BLR}, the direct product test in Subsection~\ref{subsec:DPT}. Finally, in Subsection~\ref{subsec:sic-proof} we prove the main result on the test.

We start by introducing some notation.

Given two vectors $a = (a_1,\dots,a_d),\,b = (b_1,\dots,b_d)\in[\overline{n};d]$, define
\begin{itemize}
\item  $\Delta(a,b)=\{i: a_i\neq b_i\}\subseteq [d]$;
\item the induced subcube $C_{a,b}$ is the binary cube $\F^{\Delta(a,b)}$;
\item the projection map $\rho_{a,b}:C_{a,b}\to[\overline{n};d]$ defined for $x\in C_{a,b}$ as
\begin{equation}\nonumber
\rho_{a,b}\left(x\right)_i =\begin{cases}
										a_i=b_i,\, & i\not\in \Delta(a,b);\\
										b_i,\, & i\in \Delta(a,b)\text{ and }x_i = 1;\\
										a_i,\, & i\in \Delta(a,b)\text{ and }x_i = 0;
									 \end{cases}
\end{equation}
\end{itemize}

The following test is the same as Test~\ref{test:SiC-intro} in Introduction.

\begin{figure}[h]\captionof{ctest}{Square in a Cube test. Given a query access to a function $f: [\overline{n};d]\to\F$:}\label{test:sic}
\begin{enumerate}
\item Choose $a,b\in [\overline{n};d]$ uniformly at random.
\item Choose $x,y\in C_{a,b}$ uniformly at random.
\item Query $f$ at $\rho_{a,b}(0),\rho_{a,b}(x),\rho_{a,b}(y)$ and $\rho_{a,b}(x\oplus y)$.
\item Accept iff $f(\rho_{a,b}(0))\oplus f(\rho_{a,b}(x))\oplus f(\rho_{a,b}(y))\oplus f(\rho_{a,b}(x\oplus y))=0$.
\end{enumerate}
\end{figure}

\begin{theorem}\label{thm:sic}
Suppose a function $f:[\overline{n};d]^d\to\F$ passes Test~\ref{test:sic} with probability $1-\e$ for some $\e > 0$, then $f$ is $O(\e)$-close to a tensor product.
\end{theorem}

\subsection{The BLR affinity test}\label{subsec:BLR}
The Blum-Luby-Rubinfeld linearity test was introduced in~\cite{Blum1993}, where its remarkable properties were proven. Later a simpler proof via  Fourier analysis was presented, e.g. see~\cite{BCHKS95}. Below we give a variation of this test for affine functions, see \cite[Chapter 1]{Odonnellbook}.

\begin{definition}\label{def-affine}
A function $g:\F^d \to \F$ is called affine, if there exists a set $S\subseteq [d]$ and a constant $c\in\F$ such that for every vector $x\in \F^d$
$$
g(x) = c\oplus \bigoplus _{i\in S} x_i.
$$
\end{definition}
Note that (see \cite[Exercise 1.26]{Odonnellbook}) a function $g$ is affine iff for any two vectors $x,y\in \F^d$ it satisfies
\begin{equation}\label{prop-affine}
g(0)\oplus g(x)\oplus g(y)\oplus g(x\oplus y) = 0.
\end{equation}

The BLR test implies that if a function $g:\F^d \to \F$ satisfies~(\ref{prop-affine}) with high probability, then it is close to an affine function.

\begin{figure}[h]\captionof{ctest}{The BLR affinity test. Given a query access to a function $f: \F^d\to\F$:}\label{test:blr}
\begin{enumerate}
\item Choose $x\sim \F^d$ and $y\sim \F^d$ independently and uniformly at random.
\item Query $g$ at $0,x,y$ and $x\oplus y$.
\item Accept if $g(0)\oplus g(x)\oplus g(y)\oplus g(x\oplus y) = 0$.
\end{enumerate}
\end{figure}

\begin{theorem}[\cite{Blum1993}]\label{thm:BLR} Suppose $g:\F^d\to\F$ passes the affinity test with probability $1-\e$ for some $\e > 0$.
Then $g$ is $\e$-close to being affine.
\end{theorem}

\subsection{Generalized Direct Product Test}\label{subsec:DPT}
\begin{definition} For $k,M,N_1,\ldots,N_k\in\mathbb{N}$, and $k$ functions $g_1,\dots,g_k:[N_i]\to[M]$, their direct product is the function $g:\prod_i[N_i]\to[M]^k$ denoted $g = g_1\times\dots\times g_k$ and defined as $g\left((x_1,\dots,x_k)\right) = (g_1(x_1),\dots,g_k(x_k))$. A function $g:\prod_i[N_i]\to[M]^k$, is called a direct product if there exist $k$ functions $g_1,\dots,g_k:[N_i]\to[M]$ such that $g = g_1\times \dots\times g_k$ for all $(x_1,\dots,x_k)\in \prod_i[N_i]$. 
\end{definition}
Dinur and Steurer~\cite{Dinur2014} presented a $2$-query test, very similar to Test~\ref{test:ds} below, that, with constant probability, distinguishes between direct products and functions that are far from direct product. 
\begin{figure}[h]\captionof{ctest}{Two-query test $\mathcal{T}(\alpha)$. Given a query access to a function $g: \prod_{i=1}^k[N_i]\to[M]^k$:}\label{test:ds}
\begin{itemize}
\item Choose $x\in \prod_{i=1}^k[N_i]$ uniformly.
\item For each $i$, with probability $\alpha$ set $y_i = x_i$ and add $i$ to $A$, and otherwise choose $y_i\in [N_i]$ uniformly.
\item Query $g$ at $x$ and $y$.
\item Accept iff $g(x)_A = g(y)_A$.
\end{itemize}
\end{figure}
The proof in \cite{Dinur2014} works for the special case of $N_1=\cdots=N_k$ and can easily be modified to work for the more general situation. Nevertheless, for completeness, we will rely on a newer and more general agreement theorem of \cite{DD} that directly implies what we need. 
\begin{theorem}[Generalized direct product testing theorem]\label{thm:ds} 
Let $k,M,N_1,\ldots,N_k\in\mathbb{N}$ be positive integers, and let $\e >0$. Let $g:\prod_i[N_i]\to [M]^k$ be a function that passes Test \ref{test:ds} with parameter $\alpha=0.75$ with probability at least $1-\e$.
Then there exist functions $h_i:[N_i]\to[M]$ such that 
\[\Pr_x \left[g(x)=(h_1(x),h_2(x),\ldots,h_k(x))\right] \ge 1-O(\e).\]
\end{theorem}
We will show in Section \ref{sec:DP} how to derive the above theorem from the agreement theorem of \cite{DD}.

\subsection{Proof of Theorem~\ref{thm:sic}}\label{subsec:sic-proof}

For a positive integer $D$, we denote by $\mu_{\nicefrac{2}{3}}(\F^D)$ the distribution on $\F^D$, where each coordinate, independently, is equal to $0$ with probability $1/3$ and to $1$ with probability $2/3$.

We use the following proposition in the course of the proof.

\begin{proposition}\label{prop:emptyset} Let $S\subseteq [D]$ be a set and $\chi_S: \F^D\to\F$ be the corresponding linear function, i.e., $\chi_S(x) = \bigoplus_{i\in S}x_i$. Suppose
$$
\Pr_{x\sim \mu_{\nicefrac{2}{3}}(\F^D)}\left(\chi_S(x) = 0\right) > \frac{2}{3},
$$
then $S = \emptyset$.
\end{proposition}
\begin{proof}
Consider $(-1)^{\chi_S}$. Then
$$
\Pr_{x\sim \mu_{\nicefrac{2}{3}}(\F^D)}\left(\chi_S(x) = 0\right) = \Pr_{x\sim \mu_{\nicefrac{2}{3}}(\F^D)}\left((-1)^{\chi_S(x)} = 1\right).
$$
Also the following holds
$$
\frac{1}{3} < \left| 2 \Pr_{x\sim \mu_{\nicefrac{2}{3}}(\F^D)}\left((-1)^{\chi_S(x)} = 1\right) - 1 \right| =  \left| \mathbb{E}_{x\sim \mu_{\nicefrac{2}{3}}(\F^D)} (-1)^{\chi_S(x)} \right| =
$$
$$
\left| \prod_{i\in [D]} \mathbb{E}_{x_i\sim \mu_{\nicefrac{2}{3}}(\F)} (-1)^{x_i}\right| = \left| \left(-\frac{1}{3}\right)^{|S|}\right| =  \left(\frac{1}{3}\right)^{|S|},
$$
and the statement follows.
\end{proof}

\begin{proof} (of Theorem~\ref{thm:sic}.)
Assume Test~\ref{test:sic} rejects a function $f:[\overline{n};d]\to\F$ with probability less than $\e$, i.e.,
$$
\Pr_{\substack{a,b\sim [\overline{n};d] \\ x,y\sim C_{a,b}}}
\left( f_{a,b}(0)\oplus f_{a,b}(x)\oplus f_{a,b}(y)\oplus f_{a,b}(x\oplus y) = 0\right)  >1- \e,
$$
where all distributions are uniform, and $f_{a,b}$ is a shorthand for $f\circ \rho_{a,b}$. Then there exists $a\in [\overline{n};d]$ such that
$$
\Pr_{\substack{b\sim [\overline{n};d] \\ x,y\sim C_{a,b}}}
\left( f_{a,b}(0)\oplus f_{a,b}(x)\oplus f_{a,b}(y)\oplus f_{a,b}(x\oplus y) = 0\right) > 1-\e.
$$

Note that the operations re-indexing the domain $[\overline{n};d]$\footnote{By this we mean selecting permutations $\pi_i$ on $[n_i]$ for $i=1,\dots,d$, and setting $f^{\pi_1,\dots,\pi_d}\left(x_1,\dots,x_d\right) = f\left(\pi_1(x_1),\dots,\pi_d(x_d)\right)$}, as well as {\em flipping} a function, i.e., adding the constant one function to it element-wise, preserve the distance between functions. Hence, w.l.o.g. we can assume for convenience that $a = (0,\dots,0)$ and that $f(a) = 0$. 

We write $C_b$ for $C_{a,b}$ and $f_{b}$ for $f_{a,b}$. Then for every $b\in [\overline{n};d]$,
$$
\Pr_{\substack{ x,y\sim C_{b}}}
\left( f_{b}(0)\oplus f_{b}(x)\oplus f_{b}(y)\oplus f_{b}(x\oplus y) = 0\right)=1-\e_b.
$$

The BLR theorem (Theorem~\ref{thm:BLR}) implies that for each $b\in [\overline{n};d]$ there exists a subset $S(b)\subseteq\Delta(a,b)$, such that
$$
\Pr_{\substack{ x\sim C_{b}}}
\left( f_{b}(x) = \chi_{S(b)}(x) \right)=1-\e_b.
$$
\begin{remark}
By the BLR theorem, there should be the \textquote{greater or equal to} sign instead of the equality. We assume equality for convenience.
\end{remark}

Let $F:[\overline{n};d]\to\F^d$ be a function defined as follows. For each $b\in[\overline{n};d]$, the set $S(b)\subseteq \Delta(a,b)$ can be viewed as a subset of $[d]$, since $\Delta(a,b)\subseteq [d]$. Then $F(b)$ is defined as the element of $\F^d$ corresponding to the set $S(b)$.

We now show that $F$ passes Test~\ref{test:ds} with high probability and hence is close to a direct product.

Let $b\in [\overline{n};d]$ be chosen uniformly at random, and let $b'\in[\overline{n};d]$ be chosen with respect to the following distribution $D(b)$. For each $i\in [d]$,
$$
b_i' =\begin{cases}
	b_i,\, & \text{w.p. } \nicefrac{3}{4};\\
	\text{chosen uniformly at random from } [n]\setminus\{b_i\},\, & \text{w.p. } \nicefrac{1}{4}.
	\end{cases}
$$
Note that the distribution on pairs $(b,b')$, where $b$ is chosen uniformly from $[\overline{n};d]$ and $b'$ w.r.t. $D(b)$, is equivalent to the following: for each $i\in [d]$,
\begin{equation}
\begin{cases}
	b_i = b_i'\text{ chosen uniformly from }[n],\, & \text{w.p. } \nicefrac{3}{4};\\
	b_i \neq b_i'\text{ both chosen uniformly from }[n] \, & \text{w.p. } \nicefrac{1}{4}.
\end{cases}
\end{equation}
In particular, it is symmetric in the sense that choosing $b'\sim[\overline{n};d]$ uniformly at random first, and then $b\sim D(b')$, leads to the same distribution on pairs $(b,b')$ as the one described above.

For such a pair $(b,b')$ define distribution $\mathcal{D}_{b,b'}$ on $[\overline{n};d]$ as follows. For a vector $x\sim \mathcal{D}_{b,b'}$,
$$
x_i =\begin{cases}
		 0,\, & \text{if } i\in\Delta(b,b');\\
		\begin{aligned}
		& 0,\, & \text{w.p. } \nicefrac{1}{3};\\
	    & b_i = b_i'  \, & \text{w.p. } \nicefrac{2}{3}.
		\end{aligned} & \text{if } i\not\in\Delta(b,b').
	\end{cases}
$$
Note that the distribution $\mathcal{D}_{b,b'}$ is supported on a binary cube of dimension $d-|\Delta(b,b')|$ inside $[\overline{n};d]$.
Denote
$$
\e_{b,b'} = \Pr_{x\sim \mathcal{D}_{b,b'}}\left( f(x) \neq \chi_{F(b)(x)} \right).
$$
We claim that the following holds
\begin{equation}\label{eq:exp_eps}
\e_b = \Pr_{x\sim C_{b}}\left( f(x) \neq \chi_{F(b)(x)} \right)  = \mathop{\mathbb{E}}_{b'\sim D(b)}  \e_{b,b'}.
\end{equation}
To see~(\ref{eq:exp_eps}) note that since $b$ is chosen uniformly, $b'$ is chosen w.r.t. $D(b)$, and $x\sim \mathcal{D}_{b,b'}$, the resulting distribution for $x$ is
$$
x_i =\begin{cases}
		0,\, & \text{w.p. } \nicefrac{1}{2};\\
		b_i \, & \text{w.p. } \nicefrac{1}{2},
	\end{cases}
$$
which is exactly the uniform distribution on $C_{b}$.

We now show that
\begin{equation}\label{inq:Markov}
\Pr_{\substack{b\sim [\overline{n};d] \\ b'\sim D(b)}} \left( \e_{b,b'} + \e_{b',b} > \frac{1}{3}\right) < 6\e
\end{equation}
First note that it follows from the definitions that
$$
 \mathop{\mathbb{E}}_{b\sim [\overline{n};d]} \mathop{\mathbb{E}}_{b'\sim D(b)} \e_{b,b'} =
  \mathop{\mathbb{E}}_{b\sim [\overline{n};d]} \e_{b} = \e.
$$
And by the symmetry of the distribution on pairs $(b,b')$,
$$
 \mathop{\mathbb{E}}_{b\sim [\overline{n};d]} \mathop{\mathbb{E}}_{b'\sim D(b)} \e_{b',b} =   \mathop{\mathbb{E}}_{b'\sim D(b)} \mathop{\mathbb{E}}_{b\sim [\overline{n};d]} \e_{b',b} = \e.
$$
Combined together, the previous two equations imply that
$$
 \mathop{\mathbb{E}}_{b\sim [\overline{n};d]} \mathop{\mathbb{E}}_{b'\sim D(b)} \left(\e_{b,b'}+\e_{b',b}\right) = 2\e,
$$
and by the Markov inequality, Inequality~\ref{inq:Markov} follows.
By the definition of $\e_{b,b'}$,
$$
\Pr_{x\sim \mathcal{D}_{b,b'}} \left( \chi_{F(b)}(x) = \chi_{F(b')}(x) \right) > 1 - \left(\e_{b,b'} + \e_{b',b}\right).
$$
which is equivalent to
$$
\Pr_{x\sim \mathcal{D}_{b,b'}} \left( \chi_{F(b)\Delta F(b')}(x) = 1 \right) > 1 - \left(\e_{b,b'} + \e_{b',b}\right).
$$
Proposition~\ref{prop:emptyset} implies that if $1 - \left(\e_{b,b'} + \e_{b',b}\right) > \frac{2}{3}$, then
$$
F(b)_{C_b\cap C_{b'}} = F(b')_{C_b\cap C_{b'}}.
$$

By Theorem~\ref{thm:ds}, the function $F:[\overline{n};d]\to\F^d$ is close to a direct product, i.e., there exist $d$ functions $F_1,\dots,F_d:[n]\to\F$ such that
$$
\Pr_{b\sim[\overline{n};d]} \left( F(b)=\left(F_1(b_1),\dots,F_d(b_d)\right)\right) \geq 1-O(\e).
$$

Therefore,
$$
\Pr_{b\sim[\overline{n};d]} \left( f(b)=\bigoplus_{i=1}^{d}F_i(b_i) \right) \geq 1-O(\e).
$$

\end{proof}
\section{The Shapka Test}\label{sec:shapkatest}
In this section we present a different test for whether a tensor is a tensor product. It queries the tensor at $(d+2)$ places at most, but the proof is simpler than for the previous test.  

In~\cite{Kaufman2014}, Kaufman and Lubotzky showed an interesting connection between the theory of high-dimensional expanders and property testing. Namely, they showed that $\F$-coboundary expansion of a $2$-dimensional complete simplicial complex implies testability of whether a symmetric $\F$-matrix is a tensor square of a vector. The following test is inspired by their work and in a way generalizes it. However, since the description below does not employ neither terminology nor machinery of high-dimensional expanders, we refer to~\cite{Kaufman2014} for the connection between this theory and property testing.

Given two strings $a,b\in[\overline{n};d]$, for $i\in[d]$ denote by $a_b^i\in[\overline{n};d]$ the vector which coincides with $a$ in every coordinate except for the $i$-th one, where it coincides with $b$, i.e.,
\begin{equation*}
 (a_b^i)_j =  \begin{cases}
						a_j,\, &\text{if } j\neq i;\\
						b_i,\, &\text{if } j=i.
					\end{cases}
\end{equation*}
For a string $a\in[\overline{n};d]$, and a number $x\in[n_i]$, we write $a_x^i$ for the string which is equal to $a$ in every coordinate except for the $i$-th one, where it is equal to $x$, i.e.,
$$
a_x^i = (a_1,\dots,a_{i-1},x,a_{i+1},\dots,a_d).
$$

\begin{figure}[h]\captionof{ctest}{The Shapka Test. Given a query access to a function $f: [\overline{n};d]\to\F$:}\label{test:hyp} 
\begin{enumerate}
\item Choose $a,b\in [\overline{n};d]$ uniformly at random.
\item Define the query set $Q_{a,b}\subseteq [\overline{n};d]$ to consist of $a$, $a_b^j$ for all $j\in [d]$, and also $b$ if $d$ is even.
\item Query $f$ at the elements of $Q_{a,b}$.
\item Accept iff $\bigoplus_{q\in Q_{a,b}} f(q) =0$.
\end{enumerate}
\end{figure}

\begin{remark}
Shapka is the Russian word for a winter hat (derived from Old French \textit{chape} for a \textit{cap}). The name \textit{the Shapka test} comes from the fact that the set $Q_{a,b}$ consists of the two top layers of the induced binary cube $C_{a,b}$ (and also the bottom layer if $d$ is even).
\end{remark}

\begin{theorem}\label{thm:hyp}
Suppose a function  $f:[\overline{n};d]\to\F$ passes Test~\ref{test:hyp} with probability $1-\e$ for some $\e > 0$, then $f$ is $\e$-close to a tensor product.
\end{theorem}

\begin{proof}
Let $\delta$ be the relative Hamming distance from $f$ to the subspace of direct sums, i.e., for every direct sum $g:[\overline{n};d] \to \F$ it holds that
$$
Pr_{x\sim [\overline{n};d]}\left(f(x) \neq g(x) \right) \geq \delta.
$$

For a vector $a\in [\overline{n};d]$, let us define the local view of $f$ from $a$, that is $d$ functions $f_1^a,\dots,f_d^a$, where $f_i^d:[n_i]\to \F,\,i=1,\dots,d$, that are defined as follows. For $1\leq i \leq d-1$, and $x\in [n_i]$, 
$$
f^a_i (x) = f(a_{x}^i).
$$
For $i=d$, the definition of $f^a_d:[n_d]\to\F$ depends on the parity of $d$ and goes as follows
\begin{equation}\nonumber
\begin{cases}
f^a_d(x) = f(a_{x}^d),\, & \text{if }d\text{ is odd}, \\
f^a_d(x) = f(a_{x}^d) \oplus  f(a),\, & \text{if }d\text{ is even}.
\end{cases}
\end{equation}

Given a collection of $d$ functions, $g_i:[n_i]\to \F,\, i=1,\dots,d$, recall that their direct sum is the function $g_1\oplus\dots\oplus g_d$ such that for a vector $x\in[\overline{n};d]$ the following holds
$$
g_1\oplus\dots\oplus g_d = \bigoplus_{i\in[d]} g_i(x_i).
$$
The following holds for any $[\overline{n};d]$,
\begin{equation}\label{eq:coboundary}
\left(f - f^a_1\oplus\dots\oplus f^a_d)\right)(b_1,\dots,b_d) = \bigoplus_{q\in Q_{a,b}} f(q).
\end{equation}
As $f^a_1\oplus\dots\oplus f^a_d$ is a direct sum, it is at least $\delta$-far from $f$, and hence for any $a\in [\overline{n};d]$,
\begin{equation}\label{eq:coboundary-exp}
\Pr_{b\sim [\overline{n};d]} \left( \left(f - f^a_1\oplus\dots\oplus f^a_d\right)(b) = 1 \right) \geq \delta.
\end{equation}
Assume now that $f$ fails Test~\ref{test:hyp} with probability $\e$, i.e.,
$$
\e = \Pr_{a,b\sim [\overline{n};d]} \left( \bigoplus_{q\in Q_{a,b}} f(q) = 1 \right).
$$
Combining this equality with (\ref{eq:coboundary}) and (\ref{eq:coboundary-exp}), we get the following
$$
\e = \mathop{\mathbb{E}}_{a\sim[\overline{n};d]} \Pr_{b\sim [\overline{n};d]} \left( \left(f -f_a^1\oplus\dots\oplus f_a^d\right)(b_1,\dots,b_d) = 1 \right)
\geq   \left(\mathop{\mathbb{E}}_{a\sim[\overline{n};d]} \delta \right)= \delta,
$$
which completes the proof.
\end{proof}

\section{Generalized direct product test}\label{sec:DP}

In this section we prove Theorem~\ref{thm:ds}, restated directly below, by relying on known agreement test results. \\~\\
\noindent {\bf Theorem \ref{thm:ds}~(restated)~}{\it 
Let $k,M,N_1,\ldots,N_k\in\mathbb{N}$ be positive integers, and let $\e >0$. Let $g:\prod_i[N_i]\to [M]^k$ be a function that passes Test \ref{test:ds} with parameter $\alpha=0.75$ with probability at least $1-\e$.
Then there exist functions $h_i:[N_i]\to[M]$ such that 
\[\Pr_x \left[g(x)=(h_1(x),h_2(x),\ldots,h_k(x))\right] \ge 1-O(\e).\]
}

This theorem was proven ``in spirit'' in \cite{Dinur2014} although formally that proof is written only for the case of $N_1=N_2=\cdots=N_k$. Instead of reworking the details we will rely on a newer work that generalizes the \cite{Dinur2014} paper to a broader context of agreement testing.

First, let us move from the distribution of Test \ref{test:ds} to a related distribution. It turns out that if $g$ passes one of these two-query tests with good probability then we can draw conclusions regarding its success in related tests.
\begin{figure}[h]\captionof{ctest}{Two-query test with fixed intersection size $\mathcal{T}(t)$. Given $g:\prod_i[N_i]\to[M]^k$}\label{test:ds-fixed}
\begin{itemize}
\item Choose $x \in \prod_i[N_i]$ uniformly.
\item Choose a subset $T\subset [k]$ of size $t$ uniformly.
\item Choose $y \in \prod_i[N_i]$ uniformly conditioned on $y|_T = x|_T$.
\item Accept iff $g(x)|_T=g(y)|_T$.
\end{itemize}
\end{figure}
\begin{claim}\label{claim:param} 
Suppose $g$ passes Test \ref{test:ds} with $\alpha = 0.75$ with probability $1-\e$ then it passes Test \ref{test:ds-fixed} with parameter $k/10< t < k/4$ probability $1-O(\e)$.
\end{claim}
We prove this claim later in Section \ref{subsec:param}.
Theorem \ref{thm:ds} will follow by invoking a theorem from \cite{DD} about agreement testing. 
In agreement testing the input is a collection of local functions each defined on its own small domain. The agreement test checks that whenever the small domains overlap the functions agree with each other. An agreement theorem deduces a single global function (on a domain that contains all the smaller ones) from the given local pairwise agreements. To see who are the small domains in our context let us construct the following set system. 
\begin{itemize}
\item Vertices: Let $V_1,\ldots,V_k$ be $k$ disjoint sets of vertices, $|V_i|=N_i$ and we identify $V_i$ with $[N_i]$.
\item Subsets: We have a subset for every choice of one element from each $V_i$,
\[ {\cal S}= \{ \set{v_1,\ldots,v_k}\; : \; \forall i=1,\ldots,k,\;v_i \in V_i\}. \]
There is a straightforward bijection between ${\cal S}$ and the domain of $g$, namely $\prod_i [N_i]$.
\item Local functions:
For a set $S=\set{v_1,\ldots,v_k} \in {\cal S}$ we have a local function $f_S:S\to [M]$ defined by 
\[f_S(v_i) = g(\bar v_1,\ldots,\bar v_k)_i\] 
where $\bar v_i\in [N_i]$ is associated with $v_i$ in the identification of $V_i$ and $[N_i]$.
\end{itemize}
A direct product function $g:\prod_i[N_i]\to [M]^k$ can thus be represented as a collection $\set {f_S}$ of local functions. The direct product test, Test~\ref{test:ds-fixed}, can be rephrased as Test \ref{test:a} below. 
Given $g:\prod_i[N_i]\to[M]^k$ we view it as a family of local functions $\set{f_S}$ and would like to invoke the following agreement test theorem,
\begin{figure}[h]\captionof{ctest}{Two-query test $\mathcal{T}(t)$. Given a family of local functions $\{f_S\in [M]^S\;:\; S\in {\cal S}\}$}\label{test:a}
\begin{itemize}
\item Choose a set $S_1 \in {\cal S}$ uniformly.
\item Choose a subset $T\subset S_1$ of size $t$ uniformly.
\item Choose $S_2\in {\cal S}$ uniformly conditioned on $S_2\supset T$.
\item Accept iff $f_{S_1}|_T=f_{S_2}|_T$.
\end{itemize}
\end{figure}
\begin{theorem}[{\cite[Theorem 4.4]{DD}}]\label{thm:a}
Suppose $\S$ is a collection of subsets that are top faces of a $\lambda$-one-sided $k$-partite $\frac 1 {k^3}$-high dimensional expander. Then given $\set{f_S}$ for which Test \ref{test:a} succeeds with probability $1-\e$, and assuming $t<k/4$, there exists a function $h:V_1\sqcup\cdots\sqcup V_k\to[M]$ such that 
\[ Pr_{S\in {\cal S}}[f_S = h|_S] \ge 1-O(\e).
\]
\end{theorem}
We will show in Section \ref{subsec:complete}  that we are justified to apply this theorem because our collection of subsets, also known as the ``complete multi-partite complex'', is a $\lambda$-one-sided-HDX for any $\lambda\ge 0$. Assuming this is the case, we can now take $h_i = h|_{V_i}$ and get the desired conclusion of Theorem \ref{thm:ds},
\[\Pr_x [ g(x) = (h_1(x_1),\ldots,h_k(x_k)) ] = \Pr_S [ f_S = h|_S ] \ge 1-O(\e) .
\]
\subsection{Moving between different variants of agreement tests}\label{subsec:param}

Claim \ref{claim:param} follows immediately from the following lemma, (one needs to apply the first item 3 times to get from $\alpha=0.75$ to $\alpha^2$ then $\alpha^4$ and then $\alpha^8 < 0.25$ and then item 2 once).
\begin{lemma}
Let $g:\prod_i[N_i]\to [M]^k$ be a function that passes Test \ref{test:ds} with parameter $\alpha$ with probability at least $1-\e$. Then,
\begin{itemize}
\item $g$ passes  Test \ref{test:ds} with parameter $\alpha^2$ with probability at least $1-2\e$. 
\item There exists a number $t$, $\alpha k - \sqrt k \le t \le \alpha k+\sqrt k$, such that $g$ passes Test \ref{test:ds-fixed} with parameter $t$ with probability at least $1-O(\e)$
\end{itemize}
\end{lemma}
\begin{proof}
We first prove the first item. Choosing two queries $x,y$ according to the test distribution in Test \ref{test:ds} and then another pair $x,y'$ conditioned on the first query being $x$, we get a pair $y,y'$ whose distribution is exactly as if the were chosen from Test \ref{test:ds} with parameter $\alpha^2$. Suppose $A$ was the set of indices in which $y_i$ was chosen to equal $X_i$, and suppose $A'$ was that set for the pair $x,y'$. Setting $B=A\cap A'$  it remains to notice that the event that $g(y)|_B \neq g(y')|_B$ is contained in at least one of the events $g(y)|_A \neq g(x)|_A$ or $g(y')|_{A'} \neq g(x)|_{A'}$, so its probability is at most $2\e$.

For the second item, observe that with probability $p>0.1$ the size of the set $A$ defined by the test is some $t$ such that $\alpha k - \sqrt k \le t \le \alpha k+\sqrt k$ (this follows from Hoefding's tail bound). There must be some $t$ in this range for which the failure probability of the test is at most $2\e/p$. Otherwise, even if the test succeeds with probability $1$ when $t$ is outside this range, we would still not be able to reach a sucess probability of $1-\e$ since 
\[ \Pr[fail] \ge p\cdot 2\e/p >\e 
\]
\end{proof}
\subsection{The complete multi-partite complex}\label{subsec:complete}
The collection of subsets defined in the beginning of this section gives rise to the so-called complete multi-partite simplicial complex, by downwards closing that set system. 

We wish to show that it satisfies the requirements of Theorem \ref{thm:a}. For this we briefly recall the relevant definitions. For a more comprehensive introduction to this topic we refer the reader to \cite{DD} and the references therein. 
\begin{itemize}
\item Simplicial Complex: A simplicial complex is a hypergraph that is closed downward with respect to containment. It is $(d-1)$-dimensional if the largest hyperedge has size $d$. We refer to $X(\ell)$ as the hyperedges (also called faces) of size $\ell+1$. $X(0)$ are the vertices. It is $d$-partite if the vertices are partitioned into $d$ parts, and each hyperedge in $X(d-1)$ has one vertex from each part. 
\item Link: Given a $i$-face $\sigma$, the link of $\sigma$ is the collection of faces that are disjoint from $\sigma$ and whose union belongs to $X$, 
\[
X_\sigma = \set{ \tau \in X \;:\; \tau \cap \sigma=\phi \hbox { and } \tau\cup\sigma\in X}.\]
This is a simplicial complex whose dimension is $dim(X) - |\sigma| - 1$. 
\item Distribution: Given any probability distribution on the top faces $X(d-1)$, it propagates to a distribution on the edges by selecting a top face and then a pair of vertices in it uniformly. This gives a weighted graph that is called the {\bf $1$-skeleton} of the complex.
\item HDX: A $(d-1)$-dimensional simplicial complex is a $\lambda$-one-sided HDX if for every face $\sigma\in X(t)$, $t \le d-3$, the $1$-skeleton of the link $X_\sigma$ is a $\lambda$-one-sided expander graph, meaning that the random walk Markov chain on this weighted graph has all non-trivial normalized eigenvalues at most $\lambda$. 
\end{itemize}
The complete $d$-partite complex has parameters $n_1,\ldots,n_d$ and has a vertex set $V_i$ of size $n_i$. It is defined by the following distribution over $d$-hyperedges: For each $i$ choose $x_i \in V_i$ uniformly. This gives a probability distribution on faces $\set{x_1,\ldots,x_d}$ in $X(d-1)$. The $1$-skeleton of this complex is a graph whose vertices are $V_1\sqcup\cdots\sqcup V_d$ and whose weighted edges are obtained by selecting a random hyperedge in $X(d)$ and then a random pair of vertices inside it. The link of a face in this complex is itself a complete partite complex, with fewer parts. To show that this complex is a $\lambda$-one-sided HDX it remains to prove the following lemma,
\begin{lemma}
Let $G$ be the $1$-skeleton of a complete $d$-partite complex with parameters $n_1,\ldots,n_d$. Then the normalized adjacency matrix of $G$ has one eigenvalue of $1$, eigenvalue of $0$ with multiplicity $\sum_i n_i - d$, and the remaining $(d-1)$ eigenvalues have value $-1/(d-1)$. 

In particular, except for one eigenvalue of $1$, all of $G$'s remaining eigenvalues are non-positive. 
\end{lemma}
\begin{proof}
 Let, as before, $V_i$ denote the part of vertices of size $n_i$. The distribution on edges induced by the uniform distribution on the maximal faces is as follows. For an edge $(v_i,v_j)$, where $v_i\in V_i,\,v_j\in V_j$ and $i\neq j$, its probability is equal to
 \[
  p(v_i,v_j) = p_{i,j}=\frac{1}{{d \choose 2} n_i n_j}.
 \]
Hence the transition probability of moving from the vertex $v_i$ to the vertex $v_j$ is equal to 
\[
 \frac{p_{i,j}}{\sum_{j=1,\,j\neq i}^d n_j p_{i,j}} = \frac{p_{i,j}}{2/(dn_i)} = \frac{1}{(d-1)n_j},
\]
The transition matrix is of the following form
\[
A =  \frac{1}{d-1} \begin{bmatrix}
    0 & \frac{1}{n_2} J_{n_1\times n_2} & \frac{1}{n_3} J_{n_1\times n_3} & \dots & \frac{1}{n_d} J_{n_1\times n_d} \\
    \frac{1}{n_1} J_{n_2\times n_1} & 0 & \frac{1}{n_3} J_{n_2\times n_3} & \dots & \frac{1}{n_d} J_{n_2\times n_d} \\
    \hdotsfor{5} \\
    \frac{1}{n_1} J_{n_d\times n_1} & \frac{1}{n_2} J_{n_d\times n_2}     & \frac{1}{n_3} J_{n_d\times n_3} & \dots &   0
\end{bmatrix},
\]
where $J_{n_i\times n_j}$ stands for the all-one matrix of size $n_i \times n_j$. In order to show that $A$ has a single positive eigenvalue, we use the approach developed in~\cite{EssHar80}.  First, note that the multiplicity of $0$ is $n - d$, where $n = \sum_{i=1}^d n_i$, because the matrix $A$ is of rank $n-d$. Next, note that if $f$ is an eigenfunction with eigenvalue $\lambda\neq 0$, then 
\begin{enumerate}
 \item it is constant on $V_i$ for each $i=1,\dots,d$;
 \item and 
 $$
 \lambda \alpha_i = \frac{1}{d-1} \sum_{j=1,\, j\neq i}^d  \alpha_j,
 $$
 where $\alpha_i$ is the value of $f$ on $V_i$.
\end{enumerate}
For $v\in V_i$, 
\[
 \lambda f(v) = \frac{1}{d-1}\sum_{j=1,j\neq i}^{d}\left(\frac{1}{n_j}\sum_{u\in V_j} f(u)\right).
\]
The expression on r.h.s. is the same for every $v\in V_i$, and $\lambda \neq 0$, which completes the proof of (1). To show (2), it is enough to substitute $f(u) = \alpha_j $ for $u\in V_j$ in the equality above. 

It follows from the above that the non-zero eigenvalues of $A$ are exactly the eigenvalues of the matrix
 \[
 \frac{1}{d-1}\left(J_{d\times d}-I_{d\times_d}\right),
 \]
which has eigenvalue $1$ with multiplicity $1$, and $-\frac{1}{d-1}$ with multiplicity $(d-1)$. 
\end{proof}

\section{Further Directions}\label{sec:FurDirec}
Below we present possible directions for future research.

\begin{enumerate}
\item Can the original function $f: [\overline{n};d]\to \F$ be reconstructed by a voting scheme using the Shapka Test~\ref{test:hyp}?
\item It is plausible that the Square in the Cube test~\ref{test:sic} can be analyzed by the Fourier transform approach similarly to the analysis of the BLR test.
\item Another test in the spirit of the paper is the following.
\begin{figure}[h]\captionof{ctest}{Given a query access to a function $f: [\overline{n};d]\to\F$:}\label{test:hyp} 
\begin{enumerate}
\item Choose $a,b\in [\overline{n};d]$ uniformly at random.
\item Choose $x\in C_{a,b}$ uniformly at random.
\item Query $f$ at $\rho_{a,b}(0),\rho_{a,b}(x),\rho_{a,b}(1)$ and $\rho_{a,b}(x\oplus 1)$.
\item Accept iff $f(\rho_{a,b}(0))\oplus f(\rho_{a,b}(x))\oplus f(\rho_{a,b}(1))\oplus f(\rho_{a,b}(x\oplus 1))=0$.
\end{enumerate}
\end{figure}
We conjecture that this test is also good, i.e., if a function passes the test with high probability then it is close to a tensor product.
\end{enumerate}

\section*{Acknowledgements} 
The authors would like to thank Oded Goldreich for pointing out a gap in the proof in  a previous version of this manuscript. 

The first author is supported by ERC-CoG grant number 772839. A substantial part of the work was done while the second author held a joint postdoctoral position at The Weizmann Institute and Bar-Ilan University funded by the ERC grant number 336283. Currently, the second author is supported by the SNF grant number 200020\_169106. The second author would also like to thank the Swiss Mathematical Society for travel funding related to this paper.

\bibliographystyle{alpha}

\begin{thebibliography}{00}

\bibitem[BCH{\etalchar{+}}95]{BCHKS95}
Mihir Bellare, Don Coppersmith, Johan H{\aa}stad, Marcos~A. Kiwi, and Madhu
  Sudan.
\newblock Linearity testing in characteristic two.
\newblock In {\em 36th Annual Symposium on Foundations of Computer Science,
  Milwaukee, Wisconsin, USA, 23-25 October 1995}, pages 432--441, 1995.

\bibitem[BLR93]{Blum1993}
Manuel Blum, Michael Luby, and Ronitt Rubinfeld.
\newblock Self-testing/correcting with applications to numerical problems.
\newblock {\em Journal of computer and system sciences}, 47(3):549--595, 1993.

\bibitem[DDG{\etalchar{+}}17]{David2015}
Roee David, Irit Dinur, Elazar Goldenberg, Guy Kindler, and Igor Shinkar.
\newblock Direct sum testing.
\newblock {\em {SIAM} J. Comput.}, 46(4):1336--1369, 2017.

\bibitem[DG08]{DG08}
Irit Dinur and Elazar Goldenberg.
\newblock Locally testing direct products in the low error range.
\newblock In {\em Proc. 49th IEEE Symp. on Foundations of Computer Science},
  2008.

\bibitem[DR06]{DR06}
Irit Dinur and Omer Reingold.
\newblock Assignment testers: Towards combinatorial proofs of the {PCP}
theorem.
\newblock {\em {SIAM} Journal on Computing}, 36(4):975--1024, 2006.
\newblock Special issue on Randomness and Computation.

\bibitem[DD19]{DD}
Yotam Dikstein and Irit Dinur.
\newblock Agreement testing theorems on layered set systems.
\newblock In {\em 60th Annual IEEE Symposium on
Foundations of Computer Science (FOCS)},
2019.

\bibitem[DS14]{Dinur2014}
Irit Dinur and David Steurer.
\newblock Direct product testing.
\newblock In {\em 2014 IEEE 29th Conference on Computational Complexity (CCC)},
  pages 188--196, 2014.

\bibitem[EH80]{EssHar80}
Friedrich Esser and Frank Harary. 
\newblock On the spectrum of a complete multipartite graph.
\newblock In {\em European Journal of Combinatorics}, 1(3),  211--218, 1980.
  
\bibitem[GS97]{GolSaf97}
Oded Goldreich and Shmuel Safra.
\newblock A combinatorial consistency lemma with application to proving the
  {PCP} theorem.
\newblock In {\em RANDOM: International Workshop on Randomization and
  Approximation Techniques in Computer Science}. LNCS, 1997.


\bibitem[IJK06]{ImpagliazzoJK06}
Russell Impagliazzo, Ragesh Jaiswal, and Valentine Kabanets.
\newblock Approximately listdecoding direct product codes and uniform hardness amplification.
\newblock In {\em Proc. 47th IEEE
Symp. on Foundations of Computer Science}, 187--196, 2006.
  
\bibitem[IJKW08]{ImpagliazzoJKW08}
Russell Impagliazzo, Ragesh Jaiswal, Valentine Kabanets, and Avi Wigderson.
\newblock Uniform direct product theorems: Simplified, optimized, and derandomized.
\newblock In {\em Proc. 40th
ACM Symp. on Theory of Computing}, 39(4), 1637--1665, 2008.  
  
\bibitem[IKW12]{ImpagliazzoKW12}
Russell Impagliazzo, Valentine Kabanets, and Avi Wigderson.
\newblock New direct-product testers and 2-query {PCP}s.
\newblock {\em SIAM J. Comput.}, 41(6):1722--1768, 2012.

\bibitem[KL14]{Kaufman2014}
Tali Kaufman and Alexander Lubotzky.
\newblock High dimensional expanders and property testing.
\newblock In {\em Proceedings of the 5th Conference on Innovations in
  Theoretical Computer Science}, ITCS '14, pages 501--506, New York, NY, USA,
  2014. ACM.

\bibitem[O'D14]{Odonnellbook}
Ryan O'Donnell.
\newblock {\em Analysis of Boolean Functions}.

\bibitem[STV01]{Sudan2001}
Madhu Sudan, Luca Trevisan and Salil Vadhan.
\newblock Pseudorandom Generators without the XOR Lemma.
\newblock In {\em Journal of Computer and System Sciences}, 62(2), pages 236--266, 2001.

\bibitem[T03]{Trevisan2003}
Luca Trevisan.
\newblock List-decoding using the XOR lemma.
\newblock In {\em 44th Annual IEEE Symposium on Foundations of Computer Science, 2003.} Proceedings., pages 126--135, 2003.

\bibitem[Y82]{Yao82}
A. C. Yao.
\newblock Theory and application of trapdoor functions.
\newblock {23rd Annual Symposium on Foundations of Computer Science (sfcs 1982)}, Chicago, IL, USA, 1982, pp. 80-91.

\end{thebibliography}

\end{document}